\documentclass{article}
\usepackage{graphicx}
\usepackage{spconf}
\usepackage{color}
\usepackage{float}
\usepackage{amsmath,bm,amssymb,amsfonts,amsthm,epsfig}
\usepackage{graphicx}
\usepackage{psfrag}
\usepackage[caption=false]{subfig}
\usepackage{comment}
\usepackage{url}
\usepackage{cite} 

\graphicspath{{./}}
\usepackage{epstopdf}
\usepackage{lipsum}

\newtheorem{lemma}{Lemma}
\newtheorem{proposition}{Proposition}

\newtheorem*{thm*}{Theorem}

\newtheorem{remark}{Remark}
\newtheorem*{remark*}{Remark}

\usepackage[hidelinks]{hyperref}
\usepackage[capitalize]{cleveref}
\crefname{thm}{Theorem}{Theorems}

\newcounter{opteq}

\usepackage{arydshln}
\usepackage{mathrsfs}
\usepackage{bbm}

\usepackage{xfp}

\usepackage{pifont}

\usepackage[toc,page]{appendix}

\newcommand{\txt}[1]{\text{\normalfont #1}}

\DeclareMathOperator*{\argmax}{arg\,max}
\DeclareMathOperator{\T}{\top}

\DeclareMathOperator{\rank}{rank}

\DeclareMathOperator{\diag}{diag}

\DeclareMathOperator{\tx}{t}
\DeclareMathOperator{\rx}{r}

\newcommand{\Nt}{20}
\newcommand{\Nr}{20}
\newcommand{\Q}{12}
\newcommand{\NSigma}{\fpeval{\Nt * \Nr/2}}
\newcommand{\figsize}{4.5}

\makeatletter
\newcounter{savesection}
\newcounter{apdxsection}
\renewcommand\appendix{\par
	\setcounter{savesection}{\value{section}}%
	\setcounter{section}{\value{apdxsection}}%
	\setcounter{subsection}{0}%
	\gdef\thesection{\@Alph\c@section}}
\newcommand\unappendix{\par
	\setcounter{apdxsection}{\value{section}}%
	\setcounter{section}{\value{savesection}}%
	\setcounter{subsection}{0}%
	\gdef\thesection{\@arabic\c@section}}
\makeatother

\usepackage[inline]{enumitem}
\usepackage{tikz}
\usetikzlibrary{tikzmark,calc,decorations.pathreplacing,shapes,backgrounds,patterns}
\usepackage{pgfplots} 
\pgfplotsset{compat=1.16}
\pgfplotsset{table/search path={Data},
        colormap={parula}{
            rgb255=(53,42,135)
            rgb255=(15,92,221)
            rgb255=(18,125,216)
            rgb255=(7,156,207)
            rgb255=(21,177,180)
            rgb255=(89,189,140)
            rgb255=(165,190,107)
            rgb255=(225,185,82)
            rgb255=(252,206,46)
            rgb255=(249,251,14)
        },
        every axis/.append style={
                    label style={font=\scriptsize},
                    tick label style={font=\scriptsize},  
                    title style={font=\scriptsize},
                    legend style = {font = \scriptsize},
                    },
    contour/label node code/.code={%
        \node{$m=\pgfmathprintnumber{#1}$};}
    }
\usepgfplotslibrary{fillbetween}
\usetikzlibrary{pgfplots.groupplots}
\usepackage[percent]{overpic}

\usepackage[hang,flushmargin]{footmisc}

\newcommand\blfootnote[1]{%
  \begingroup
  \renewcommand\thefootnote{}\footnote{#1}%
  \addtocounter{footnote}{-1}%
  \endgroup
}

\makeatletter
\def\tikz@auto@anchor{%
    \pgfmathtruncatemacro\angle{atan2(\pgf@x,\pgf@y)-90}
    \edef\tikz@anchor{\angle}%
}
\makeatother

\expandafter\def\expandafter\normalsize\expandafter{%
    \normalsize%
    \setlength\abovedisplayskip{4pt}%
    \setlength\belowdisplayskip{4pt}%
}

\title{Sparse array sensor selection in ISAC with identifiability guarantees}
\name{\begin{tabular}{c} Robin Rajam\"{a}ki$^\ddagger$,
		~Piya Pal$^\ast$
		\end{tabular}}
\address{$^\ddagger$Department of Information and Communications Engineering, Aalto University, Finland\\ $^\ast$Department of Electrical and Computer Engineering, University of California San Diego, USA}

\begin{document}
	\setlength\belowcaptionskip{-15ex}
	%
	\maketitle
	\begin{abstract}	
     This paper investigates array geometry and waveform design for integrated sensing and communications (ISAC) employing sensor selection. We consider ISAC via index modulation, where various subsets of transmit (Tx) sensors are used for both communications and monostatic active sensing. The set of Tx subarrays make up a codebook, whose cardinality we maximize (for communications) subject to guaranteeing a desired target identifiability (for sensing). To characterize the size of this novel optimal codebook, we derive first upper and lower bounds, which are tight in case of the canonical uniform linear array (ULA) and any nonredundant array. We show that the ULA achieves a large codebook---comparable to the size of the conventional unconstrained case---as satisfying the identifiability constraint only requires including two specific sensors in each Tx subarray (codeword). In contrast, nonredundant arrays, which have the largest identifiability for a given number of physical sensors, only have a single admissible codeword, rendering them ineffectual for communications via sensor selection alone. The results serve as a step towards an analytical understanding of the limits of sensor selection in ISAC and the fundamental trade-offs therein.
    \blfootnote{This work was supported in part by projects Business Finland 6G-ISAC, Research Council of Finland FUN-ISAC (359094), EU Horizon INSTINCT (101139161), and grant ONR N000141912227.}
	\end{abstract}
		\begin{keywords}
			Index modulation, ISAC, sensor selection, sparse arrays, redundancy, waveform design, MIMO
		\end{keywords}
	\setlength{\textfloatsep}{7pt}

	\section{Introduction}

ISAC is envisioned to be a core technology of 6G and beyond wireless systems \cite{liu2023integrated}. In multiple-input multiple output (MIMO) ISAC, two key resources shared by the sensing and communications functionalities are the \emph{transmit} waveform and sensor array geometry, which 
impact, e.g., link reliability, throughput, 
and 
spatial resolution. 
Jointly harnessing the waveform and array geometry in an optimal manner becomes key to reap these advantages, which typically improve with an increasing 
number of sensors. However, the number of 
digital spatial channels is often limited by 
cost and power consumption constraints, as in, for instance, autonomous sensing \cite{ma2020joint,tabrikian2021cognitive} or 
millimeter wave communications \cite{mendezrial2016hybrid}. 
A common solution is to employ fewer digital channels than sensors, possibly combined with \emph{sensor selection} at the transmitter (Tx) or receiver (Rx) \cite{mendezrial2016hybrid,ahmed2020antenna,tabrikian2021cognitive,wang2023joint,xu2023abandwidth,sankar2024sparse,vanderwerf2024receiver}. 
An application of Tx sensor selection that has recently experienced a surge of renewed research interest 
stems from the family of so-called \emph{index modulation} techniques \cite{hassanien2016signaling,wang2019dual,ma2021spatial,xu2023hybrid,elbir2024index}, 
and is known as \emph{spatial modulation} in classical MIMO communications \cite{direnzo2014spatial}. Here, Tx sensor subsets (subarrays) constitute information bearing codewords, which in the context of ISAC 
are \emph{simultaneously} used for sensing by launching independent (typically orthogonal) waveforms from the selected transmit elements. This 
intimate coupling between communications and sensing 
manifests as a \emph{dual-function codebook} design problem, which is the focus of this paper. 
To tackle this challenging problem, we leverage 
recent 
advances in 
\emph{sparse sensor arrays}, which enjoy provable advantages over conventional uniform array geometries, including enhanced 
identifiabilty and resolution 
\cite{sarangi2023superresolution,amin2024sparsearrays}.

Despite active research into ISAC and sensor selection, important fundamental questions still lack sharp answers, such as 
\emph{which} (sparse) array geometries are suitable for ISAC, and \emph{how many} such configurations exist? 
This paper provides partial answers to these questions by focusing on Tx sensor selection in monostatic MIMO ISAC systems employing index modulation. 
We present first analytical results on 
ISAC codebook design with \emph{identifiability guarantees}. Specifically, we formulate a novel optimization problem for maximizing the size of the ISAC 
codebook, subject to guaranteeing a desired target identifiablity for a given number of Tx and Rx sensors. The Rx array is fixed for all codewords, which correspond to different Tx sensor subsets of equal cardinality. We fully characterize the solution in case of the 
ULA and nonredundant array, showing that the optimal codebook contains many admissible codewords (mappings) in the former case 
and only one in the latter. We also derive 
generally applicable upper and lower bounds on the optimal codebook size. These preliminary results shed light on fundamental trade-offs between 
communications and sensing in ISAC, 
paving the way towards a more complete analytical understanding thereof. 

	
    \section{Signal model and background}
    Consider a narrowband MIMO 
    ISAC system, where a base station (BS) with collocated 
    Tx and Rx 
    arrays simultaneously senses the environment and communicates with a single $M$-antenna user equipment (UE). Crucially, the BS uses the \emph{same} Tx array geometry $\mathbb{D}_{\tx}$ and spatio-temporal waveform matrix $\bm{S}\in\mathbb{C}^{T\times N_{\tx}}$ for \emph{both} (active) sensing and (downlink) communications. Here, $T$ denotes the length of the waveform
    and $N_{\tx}=|\mathbb{D}_{\tx}|$ the number of Tx sensors. 
    Similarly, the Rx array geometry at the BS, $\mathbb{D}_{\rx}\subset \mathbb{N}$ with $|\mathbb{D}_{\rx}|=N_{\rx}$, is assumed one-dimensional, collinear and collocated with the Tx array.
    
    The \emph{downlink} received signal at the UE is represented by
    \begin{align*}
    	\bm{Z}=\bm{H}\bm{S}^{\T}+\bm{W},
    \end{align*}
    where $\bm{H}\in\mathbb{C}^{M\times N_{\tx}}$ is the channel matrix between the BS and UE, and $\bm{W}\in\mathbb{C}^{M\times T}$ is a noise matrix. We assume that $\bm{H}$ is \emph{known} to the UE. 
    The \emph{communications task} of the BS is to transmit information to the UE 
    by designing the set from which $\bm{S}$ is drawn; the so-called ``codebook'' (defined formally in \cref{sec:im-isac_codebook}). The task of the UE is to decode $\bm{S}$ given measurement $\bm{Z}$ and channel $\bm{H}$, or an estimate thereof.
    
    \emph{Dual-function} waveform $\bm{S}$ is also used for sensing by the BS. Assuming $K$ far field scatterers with angular directions $\bm{\theta}\in[-\tfrac{\pi}{2},\tfrac{\pi}{2})^K$ and scattering coefficients $\bm{\gamma}\in\mathbb{C}^K$, the $N_{\rx}\!\times\!T$ backscattered signal observed at the BS assumes the form \cite{li2009mimo} 
    \begin{align}
    	\bm{Y}=\bm{A}_{\rx}(\bm{\theta})\diag(\bm{\gamma})\bm{A}_{\tx}^{\T}(\bm{\theta})\bm{S}^{\T}+\bm{N}, \label{eq:signal_sensing}
    \end{align}
    where $\bm{A}_{\rx}(\bm{\theta})\in\mathbb{C}^{N_{\rx}\times K}$ and $\bm{A}_{\tx}(\bm{\theta})\in\mathbb{C}^{N_{\tx}\times K}$ denote the manifold matrices of the Rx and Tx arrays, respectively, and $\bm{N}\in\mathbb{C}^{N_{\rx}\times T}$ is a matrix of additive noise. 
    The \emph{sensing task} of the BS is to estimate $\bm{\theta}$ given measurement $\bm{Y}$, waveform matrix $\bm{S}$, and knowledge of the array manifolds $\bm{A}_{\rx}(\cdot)$, $\bm{A}_{\tx}(\cdot)$.\vspace{-.2cm}
    
    \subsection{Communications via Tx sensor selection in ISAC}\label{sec:im_background}
    
    The BS communicates with the UE via sensor index modulation (generalized spatial modulation \cite{direnzo2014spatial}), where different codewords correspond to different Tx subarrays. Such communications schemes find applications in both conventional MIMO communications \cite{direnzo2014spatial} and modern ISAC systems \cite{wang2019dual,ma2021spatial,xu2023hybrid} where constraints on power consumption, cost, and decoding complexity can be alleviated by employing fewer RF chains than antenna elements at the transmitter. In effect, $\bm{S}$ becomes a ``sensor selection'' waveform matrix of the form
    \begin{align}
    	\bm{S}=\bm{U}\bm{J}_{\mathbb{S}}, \label{eq:sens_sel_wf}
    \end{align} 
   where $\mathbb{S}\!=\!\{g[i]\}_{i}\!\subseteq\!\mathbb{D}_{\tx}\!=\!\{d_{\tx}[n]\}_n$ is a Tx subarray and $\bm{J}_{\mathbb{S}}\!\in\!\{0,1\}^{|\mathbb{S}|\times N_{\tx}}$ a row selection matrix with $(i,n)$th entry
    \begin{align}
    	[\bm{J}_{\mathbb{S}}]_{i,n} \triangleq \begin{cases}
    		1,&\txt{if } g[i]=d_{\tx}[n]\\
    		0,&\txt{otherwise}.
    	\end{cases}\label{eq:selection_mat}
    \end{align}
    We assume that $\bm{U}\in\mathbb{C}^{T\times |\mathbb{S}|}$ is a \emph{fixed} full column rank matrix, \emph{known} to the UE.\footnote{In general, $\bm{U}$ may be unknown or only partially known to the UE, depending on whether any encoding is employed beyond sensor selection \cite{hassanien2016signaling}.} Typically, the columns of $\bm{U}$ are orthogonal, representing independent waveforms launched from the selected Tx sensors \cite{wang2019dual}. 
    Hence, 
    $\mathbb{S}$ is the information-bearing quantity 
    \emph{unknown} to the UE, 
    whose task is to detect 
    $\mathbb{S}$ belonging to some finite \emph{codebook} $\mathcal{C}$ of Tx subarrays. 
    
    The choice of $\mathcal{C}$ crucially impacts 
    communications performance, including 
    achievable rate and decoding complexity. 
    In particular, 
    let $\mathcal{C}=\mathcal{C}(Q,\mathbb{D}_{\tx})$ denote an arbitrary codebook based on selecting $Q$ out of $N_{\tx}$ sensors from Tx array $\mathbb{D}_{\tx}$. 
    Denoting the set of \emph{all} $Q$-sensor subsets of $\mathbb{D}_{\tx}$ by
    \begin{align}
    	\mathcal{C}^\txt{u}(Q,\mathbb{D}_{\tx})
    	\triangleq 
    	\{\mathbb{S}\subseteq\mathbb{D}_{\tx}: |\mathbb{S}|=Q\},\label{eq:codebook_unconstrained}
    \end{align}
    it is clear that $\mathcal{C}(Q,\mathbb{D}_{\tx})\!\subseteq\!\mathcal{C}^\txt{u}(Q,\mathbb{D}_{\tx})$ and 
    thus $|\mathcal{C}(Q,\mathbb{D}_{\tx})|\!\leq\!|\mathcal{C}^\txt{u}(Q,\mathbb{D}_{\tx})|\!=\!\binom{N_{\tx}}{Q}$. 
    From a communications perspective, directly using codebook $\mathcal{C}^\txt{u}(Q,\mathbb{D}_{\tx})$ would maximize the constellation size and thereby achievable rate at high signal-to-noise-ratios (SNRs), or simply yield maximally many codewords to choose from. 
    From a sensing perspective, however, all codewords in $\mathcal{C}^\txt{u}(Q,\mathbb{D}_{\tx})$ may \emph{not} be equally desirable. For example, some may yield unsatisfactory 
    Tx or joint Tx-Rx beampatterns \cite{wang2019dual,ma2021spatial}. The question arises how large the 
    codebook can be to meet a desired level of sensing performance? 
    We explore this question by focusing on \emph{identifiability} of $\bm{\theta}$ in \eqref{eq:signal_sensing} as a key performance indicator (KPI) for sensing. This \emph{fundamental} KPI has not been considered before in ISAC 
    codebook design. 
    As we will see, the maximal codebook size can be analytically characterized as a function of the number of identifiable targets, providing insight into trade-offs and optimal array/waveform designs in ISAC.\vspace{-.2cm}
    
    \subsection{Sum co-array and identifiability in active sensing}
    The number of targets that can be 
    identified by a 
    monostatic 
    active sensing system 
    depends on the so-called sum co-array:
    \begin{align*}
        \mathbb{D}_{\tx}+\mathbb{D}_{\rx}=\{d_{\tx}+d_{\rx}: d_{\tx}\in\mathbb{D}_{\tx}, d_{\rx}\in\mathbb{D}_{\rx}\}.
    \end{align*}
    It is well known that up to $K\leq |\mathbb{D}_{\tx}+\mathbb{D}_{\rx}|/2$ targets can be uniquely identified from \eqref{eq:signal_sensing}, in absence of noise, by appropriate waveform and array design \cite[Ch.~1]{li2009mimo}. 
    A sufficient condition for achieving this upper bound is that the sum co-array is \emph{contiguous} and waveform matrix $\bm{S}$ has full column rank ($\rank(\bm{S})=N_{\tx}$). A contiguous sum co-array satisfies $\mathbb{D}_{\tx}+\mathbb{D}_{\rx}=\mathbb{U}_{N_\Sigma}$, where $N_\Sigma\triangleq |\mathbb{D}_{\tx}+\mathbb{D}_{\rx}|$ denotes the number of sum co-array elements, and $\mathbb{U}_{N}\triangleq \{0,1,\ldots,N-1\}$ the set of nonnegative integers smaller than $N$, i.e., the normalized sensor positions of a ULA in units of \emph{half carrier wavelengths}. However, when $\bm{S}$ is column-rank-deficient ($\rank(\bm{S})<N_{\tx}$), a contiguous sum co-array is no longer sufficient, but judicious waveform design becomes necessary for achieving the maximum identifiability $N_\Sigma/2$ facilitated by the array geometry \cite{rajamaki2023importance}. This is also the case when employing Tx sensor selection, since any $\bm{S}$ in \eqref{eq:sens_sel_wf} satisfies $\rank(\bm{S})\leq |\mathbb{S}|$, where $|\mathbb{S}|<N_{\tx}$ holds when $\mathbb{S}$ is a strict subset of $\mathbb{D}_{\tx}$. 
    Indeed, when the sum co-array is contiguous, 
    it can be shown that 
    a sensor selection waveform (of rank $|\mathbb{S}|$) corresponding to Tx subset $\mathbb{S}\subseteq \mathbb{D}_{\tx}$ achieves maximal identifiability if and only if the sum set of $\mathbb{S}$ with the Rx array $\mathbb{D}_{\rx}$ is also contiguous \cite{rajamaki2023arrayinformed}, that is,
    \begin{align}
    \mathbb{S}+\mathbb{D}_{\rx}=\mathbb{D}_{\tx}+\mathbb{D}_{\rx}=\mathbb{U}_{N_\Sigma}.\label{eq:sum_set}
    \end{align}
    Next, we 
    investigate how many such 
    $\mathbb{S}$ exist at most.
    
    \section{Identifiability-maximizing codebook: Optimal Tx sensor selection for ISAC}\label{sec:im-isac_codebook}
    We propose constraining the sum set of all codewords in the ISAC codebook to guarantee maximum identifiability. 
    First, we define the set of $Q$-sensor Tx subarrays whose sum sets equal that of a given physical Tx-Rx array pair $(\mathbb{D}_{\tx},\mathbb{D}_{\rx})$: 
    \begin{align*}
    	\mathcal{C}^\txt{c}(Q,\mathbb{D}_{\tx},\mathbb{D}_{\rx})
    	\triangleq \{\mathbb{S}\subseteq\mathbb{D}_{\tx}: |\mathbb{S}|=Q,\  \mathbb{S}+\mathbb{D}_{\rx}=
    	\mathbb{D}_{\tx}+\mathbb{D}_{\rx}
    	\}.
    \end{align*}
  	Observe that the Rx array geometry $\mathbb{D}_{\rx}$ is \emph{fixed} for all codewords (Tx subarrays $\mathbb{S}$). 
  	We now wish to optimize codebook $\mathcal{C}^\txt{c}(Q,\mathbb{D}_{\tx},\mathbb{D}_{\rx})$ with respect to physical array geometries $\mathbb{D}_{\tx},\mathbb{D}_{\rx}$, given parameter tuple $(Q,N_{\tx},N_{\rx},N_\Sigma)$. Herein, we consider maximizing the size of $\mathcal{C}^\txt{c}$ subject to each codeword $\mathbb{S}\in\mathcal{C}^\txt{c}$ 
    achieving a contiguous 
    sum set $\mathbb{S}+\mathbb{D}_{\rx}=\mathbb{U}_{N_\Sigma}$:
    \begin{align*}
    	(\mathcal{C}^\star, \mathbb{D}_{\tx}^\star,\mathbb{D}_{\rx}^\star)\!\triangleq\!\argmax_{\mathcal{C},\mathbb{D}_{\tx},\mathbb{D}_{\rx}}\{|\mathcal{C}|:\ &\mathcal{C}\!=\!\mathcal{C}^\txt{c}(Q,\mathbb{D}_{\tx},\mathbb{D}_{\rx}), |\mathbb{D}_{\tx}|\!=\!N_{\tx},\\[-2ex] &|\mathbb{D}_{\rx}|=N_{\rx}, \mathbb{D}_{\tx}+\mathbb{D}_{\rx}=\mathbb{U}_{N_\Sigma}\}.
    \end{align*}
    We refer to $\mathcal{C}^\star\subseteq\mathcal{C}^\txt{u}(Q,\mathbb{D}_{\tx}^\star)$ as the \emph{identifiability-maximizing} ISAC (IM-ISAC) codebook, or \emph{optimal} codebook for short. Specifically, $\mathcal{C}^\star$ represents a set of $Q$-sensor Tx subarrays yielding a contiguous sum set of size $N_\Sigma$. Importantly, among all codebooks with this property, $\mathcal{C}^\star$ has the \emph{largest} cardinality. Note that each codeword in $\mathcal{C}^\star$ encodes $\lfloor\log_2|\mathcal{C}^\star| \rfloor$ bits and guarantees the identifiability of $N_\Sigma/2$ targets. Optimal codebook $\mathcal{C}^\star$, and the associated Tx array $\mathbb{D}_{\tx}^\star$ and Rx array $\mathbb{D}_{\rx}^\star$, need not be unique. 
    Finding (any) $\mathcal{C}^\star$ is nevertheless a challenging combinatorial problem. We will therefore instead seek bounds on $|\mathcal{C}^\star|$ for admissible tuples $(Q,N_{\tx},N_{\rx},N_\Sigma)$.\vspace{-.2cm}
    
    \section{Bounds on size of optimal codebook:\\ Preliminary insights  into ISAC trade-off}
    
    
    We start by specifying for which tuples $(Q,N_{\tx},N_{\rx},N_\Sigma)$ set $\mathcal{C}^\star$ is \emph{nonempty}. Note that the range of values that such $Q$ and $N_\Sigma$ can take depends on (free) parameters $N_{\tx}, N_{\rx}\in\mathbb{N}_+$. 
    
    Firstly, the number of sum co-array elements can be shown to satisfy $N_\Sigma\in[N_{\tx}+N_{\rx}-1,N_{\tx}N_{\rx}]$, 
	where the lower and upper bounds correspond to maximally and minimally redundant array configurations, respectively. However, 
    \begin{align}
    N_\Sigma \in [N_{\tx}+N_{\rx}-1, Q N_{\rx}] \label{eq:NSigma_range}
    \end{align}
    must actually hold for any $Q\in[1,N_{\tx}]$ such that Tx subarray $\mathbb{S}\subseteq\mathbb{D}_{\rx}$ has a contiguous sum set $\mathbb{S}+\mathbb{D}_{\rx}$ equal to the sum co-array $\mathbb{D}_{\tx}+\mathbb{D}_{\rx}$, since $|\mathbb{S}+\mathbb{D}_{\rx}|\leq|\mathbb{S}| |\mathbb{D}_{\rx}|= QN_{\rx}\leq N_{\tx}N_{\rx}$.
    
    Similarly, the minimum value of $Q$ is given by
    \begin{align}
    	L\triangleq \lceil N_\Sigma/N_{\rx}\rceil.\label{eq:L}
    \end{align}
    Indeed, the contiguous sum set constraint $\mathbb{S}+\mathbb{D}_{\rx}=\mathbb{U}_{N_\Sigma}$ implies that $|\mathbb{S}||\mathbb{D}_{\rx}|\geq|\mathbb{S}+\mathbb{D}_{\rx}|= N_\Sigma$, i.e., $Q=|\mathbb{S}|\geq N_\Sigma/N_{\rx}$. 
    Since $Q$ is an integer and $\mathbb{S}\subseteq \mathbb{D}_{\tx}$, we have
    \begin{align}
    	Q\in[L,N_{\tx}]. \label{eq:Qrange}
    \end{align}
    Interestingly, any sensor selection waveform in \eqref{eq:sens_sel_wf} with $Q=L$ is a special case of a family of \emph{optimal} rank waveforms. Specifically, $\rank(\bm{S})=L$ is a lower bound on the rank of waveforms that can fully leverage the sum co-array in the sense of identifying the maximum number of $N_\Sigma/2$ targets \cite{rajamaki2023importance}. Whether this lower bound is achievable in the case of a given redundant ($N_\Sigma<N_{\tx}N_{\rx}$) array depends not only on the array configuration itself, but also on whether the chosen waveform $\bm{S}$ is \emph{matched} to the array geometry \cite{rajamaki2023importance}. In the case of sensor selection waveforms, this matching---assuming a contiguous sum co-array---can be verified to be equivalent to satisfying \eqref{eq:sum_set} 
    and $\rank(\bm{U})=|\mathbb{S}|$ in \eqref{eq:sens_sel_wf} \cite{rajamaki2023arrayinformed}.


    Any positive tuple $(Q,N_{\tx},N_{\rx},N_\Sigma)$ is \emph{admissible}, i.e., $\mathcal{C}^\star(Q,N_{\tx},N_{\rx},N_\Sigma)$ is nonempty, if 
    \labelcref{eq:Qrange,eq:NSigma_range} are satisfied.\vspace{-.2cm}

    \subsection{Upper bound on $|\mathcal{C}^\star|$: Tx sensors on edges necessary}
     For any admissible $(Q,N_{\tx},N_{\rx},N_\Sigma)$ tuple, the upper bound $|\mathcal{C}^\star(Q,N_{\tx},N_{\rx},N_\Sigma)|\leq \binom{N_{\tx}}{Q}$ clearly holds. A less trivial upper bound follows from noting that the \emph{extremal}
     Tx sensors 
     must be included in any Tx subset $\mathbb{S}$ for \eqref{eq:sum_set}  
     to hold.
    \begin{lemma}[Upper bound]\label{thm:ub}
    	For any admissible $(Q,N_{\tx},N_{\rx},N_\Sigma)$,
    	\begin{align}
    		|\mathcal{C}^\star(Q,N_{\tx},N_{\rx},N_\Sigma)|\leq \binom{N_{\tx}-2}{Q-2}.\label{eq:ub}
    	\end{align}
    \end{lemma}
    \begin{proof}
    	The proof follows by contradiction. Suppose elements $l=\min\mathbb{D}_{\tx}$ and $u=\max\mathbb{D}_{\tx}$ are not included in $\mathbb{S}$. Then $\mathbb{S}+\mathbb{D}_{\rx}\subseteq \mathbb{D}_{\tx}\setminus\{l,u\}+\mathbb{D}_{\rx}\subseteq(\mathbb{D}_{\tx}+\mathbb{D}_{\rx})\setminus\{l+\min\mathbb{D}_{\rx},u+\max\mathbb{D}_{\rx}\}\subset \mathbb{D}_{\tx}+\mathbb{D}_{\rx}$. Hence, $\mathbb{S}+\mathbb{D}_{\rx}= \mathbb{D}_{\tx}+\mathbb{D}_{\rx}$ only if $\mathbb{S}\supseteq \{l,u\}$. This implies that at least two elements of $\mathbb{S}$ must be fixed, which leaves at most $\binom{N_{\tx}-2}{Q-2}$ choices for subset $\mathbb{S}$ satisfying $\mathbb{S}+\mathbb{D}_{\rx}= \mathbb{D}_{\tx}+\mathbb{D}_{\rx}$, regardless of $\mathbb{D}_{\tx},\mathbb{D}_{\rx}$.
    \end{proof}
     Surprisingly, the simple upper bound in \cref{thm:ub} is \emph{tight} for two canonical array geometries illustrated in \cref{fig:ula_nonredundant}: 
    \begin{enumerate*}[label=(\alph*)]
    	\item ULA Tx and Rx arrays (of appropriate size), and 
    	\item \emph{any} nonredundant array.
    \end{enumerate*}
    In these two key cases, we have thus \emph{fully} characterized the size of the optimal IM-ISAC codebook.\vspace{-.2cm}
    
    \subsection{Full characterization of ULA \& nonredundant array}
    
    \begin{proposition}[ULA and nonredundant array]\label{thm:ula_nonred}
    	If $N_\Sigma=N_{\tx}+N_{\rx}-1$ (ULA) and $N_{\tx}\leq N_{\rx}+1$, then for any $Q\in[2,N_{\tx}]$ 
    	\begin{align}
    		|\mathcal{C}^\star(Q,N_{\tx},N_{\rx},N_\Sigma)|=\binom{N_{\tx}-2}{Q-2}.\label{eq:card_ula}
    	\end{align}
    	If $N_\Sigma\!=\!N_{\rx}N_{\tx}$ (nonredundant array), then $Q\!=\!N_{\tx}$ and
    	\begin{align}
    		|\mathcal{C}^\star(Q,N_{\tx},N_{\rx},N_\Sigma)|=1.\label{eq:card_nonr}
    	\end{align}
    \end{proposition}
    \begin{proof}
        ULA:
       	Let $Q\geq 2$ and $\mathbb{S}\supseteq \{0,N_{\tx}-1\}$, and recall that $\mathbb{D}_{\tx}=\mathbb{U}_{N_{\tx}}$ and $\mathbb{D}_{\rx}=\mathbb{U}_{N_{\rx}}$. Hence, $\mathbb{S}+\mathbb{D}_{\rx}\supseteq \{0,N_{\tx}-1\}+\mathbb{U}_{N_{\rx}}=\mathbb{U}_{N_{\rx}}\cup(\mathbb{U}_{N_{\rx}}+N_{\tx}-1)=\mathbb{U}_{N_{\rx}+N_{\tx}-1}$, where the last equality follows from assumption $N_{\tx}\leq N_{\rx}+1$. Hence, $\mathbb{D}_{\tx}+\mathbb{D}_{\rx}\supseteq \mathbb{S}+\mathbb{D}_{\rx}\supseteq \mathbb{U}_{N_{\rx}+N_{\tx}-1}=\mathbb{D}_{\tx}+\mathbb{D}_{\rx}$, i.e., $\mathbb{S}+\mathbb{D}_{\rx}=\mathbb{D}_{\tx}+\mathbb{D}_{\rx}$. 
        The remaining $Q-2$ sensors of $\mathbb{S}$ can be freely chosen from $\mathbb{D}_{\tx}\setminus\{0,N_{\tx}-1\}$ without affecting the sum set. Since there are $\binom{N_{\tx}-2}{Q-2}$ such unique choices of $\mathbb{S}$, \eqref{eq:ub} is tight.
        
        Nonredundant array: Note that 
        $N_\Sigma=N_{\rx}N_{\tx}\implies Q=N_{\tx}$ by \labelcref{eq:Qrange}. 
        Substitution into \labelcref{eq:ub}, along the fact that $\mathbb{S}=\mathbb{D}_{\tx}$ (implying that $ |\mathcal{C}^\star|\geq 1$), then yields \eqref{eq:card_nonr}.
    \end{proof}

    The ULA achieves the maximum IM-ISAC codebook size $\binom{N_{\tx}-2}{Q-2}$ by \cref{thm:ula_nonred} 
    and hence enjoys an abundance of codeword choices. 
    This array geometry is of great interest 
    due to its ubiquity 
    in the ISAC literature and beyond. 
    The large codebook of the ULA 
    follows from the \emph{sufficiency}, not just necessity, of
    including the outermost sensors in any Tx subarray ($\mathbb{S}\supseteq\{\min\mathbb{D}_{\tx},\max\mathbb{D}_{\tx}\}$) 
    for the respective sum set to be contiguous ($\mathbb{S}+\mathbb{D}_{\rx}=\mathbb{U}_{N_\Sigma}$)---under mild conditions ($N_{\tx}\leq N_{\rx}+1$). 
    This fact, revealed by 
    the proof of \cref{thm:ula_nonred}, 
    is illustrated in \cref{fig:sens_sel_ULA}. 
    For nonredundant arrays, $\mathcal{C}^\star$ is nonempty only if $Q=N_{\tx}$, which corresponds to selecting \emph{all} Tx sensors. Hence, communication solely via sensor selection is \emph{not} possible when a sum co-array of (maximum) size $N_\Sigma=N_{\tx}N_{\rx}$ is desired, i.e., \emph{redundancy is needed}.
    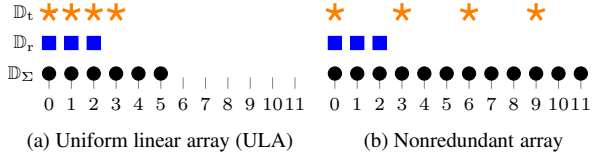
\begin{figure}
		\centering
		\newcommand{\Na}{3}   
		\newcommand{\Nb}{4}
		\newcommand{\wfac}{.71}
        \renewcommand{\figsize}{7}
		\subfloat[Uniform linear array (ULA)]{%
			\begin{tikzpicture} 
				\begin{axis}[width=\wfac*\figsize cm,height=\figsize*0.4 cm,ytick={-1,0,1},yticklabels={$\mathbb{D}_{\Sigma}$,$\mathbb{D}_{\rx}$,$\mathbb{D}_{\tx}$},xmin=-0.2,xmax=\Na*\Nb-1+0.2,ymin=-1.5,ymax=1.5,xtick={0,1,...,\Na*\Nb-1},title style={yshift=0 pt},xticklabel shift = 0 pt,xtick pos=bottom,ytick pos=left,axis line style={draw=none},]
					\addplot[blue,only marks,mark=square*,mark size=2.5] expression [domain=0:(\Na-1), samples=\Na] {0};
					\addplot[orange,only marks,mark=star,mark size=3.5, very thick] expression [domain=0:(\Nb-1), samples=\Nb] {1};
					\addplot[only marks,mark=*,mark size=2.5] expression [domain=0:(\Na+\Nb-2), samples=\Na+\Nb-1] {-1};
				\end{axis}
			\end{tikzpicture}
		}
		\subfloat[Nonredundant array]{
			\begin{tikzpicture} 
				\begin{axis}[width=\wfac*\figsize cm ,height=\figsize*0.4 cm,ytick={},yticklabels={},xmin=-0.2,xmax=\Na*\Nb-1+0.2,ymin=-1.5,ymax=1.5,xtick={0,1,...,\Na*\Nb-1},title style={yshift=0 pt},xticklabel shift = 0 pt,xtick pos=bottom,axis line style={draw=none},ymajorticks=false]
					\addplot[blue,only marks,mark=square*,mark size=2.5] expression [domain=0:(\Na-1), samples=\Na] {0};
					\addplot[orange,only marks,mark=star,mark size=3.5,very thick] expression [samples at ={0,\Na,...,\Na*(\Nb-1)}] {1};	 				
					\addplot[only marks,mark=*,mark size=2.5] expression [domain=0:(\Na*\Nb-1), samples=\Na*\Nb] {-1};
				\end{axis}
			\end{tikzpicture}
		}\vspace{-.2cm}
        \caption{Example of ULA and nonredundant array geometry. \cref{thm:ula_nonred} fully characterizes $|\mathcal{C}^\star|$ for these two cases.}\label{fig:ula_nonredundant}
	\end{figure}

    \begin{figure}
        \centering
        \newcommand{\Na}{4}   
		\newcommand{\Nb}{5}
        \newcommand{\Nbb}{3}   
		\newcommand{\D}{4}
        \newcommand{\Lmax}{11}
		\newcommand{\wfac}{.71}
        \renewcommand{\figsize}{7}
        \subfloat[$\mathbb{S}=\{0,1\}\subset \mathbb{D}_{\tx}=\mathbb{U}_5$]{\label{fig:sens_sel_ULA}
		\begin{tikzpicture} 
			\begin{axis}[width={\wfac*\figsize cm},height={\figsize*0.4 cm},
            ytick={-1,0,1},yticklabels={$\mathbb{S}+\mathbb{D}_{\rx}$,$\mathbb{D}_{\rx}$,$\mathbb{S}$},
            xmin=-0.2,xmax=\Lmax+0.2,ymin=-1.5,ymax=1.5,xtick={0,1,...,\Lmax},title style={yshift=0 pt},xticklabel shift = 0 pt,xtick pos=bottom,ytick pos=left,axis line style={draw=none}]
				\addplot[blue,only marks,mark=square*,mark size=2.5] expression [domain=0:(\Na-1), samples=\Na] {0};
				\addplot[orange,only marks,mark=star,mark size=3.5, very thick,opacity=0.2] expression [domain=0:(\Nb-1), samples=\Nb] {1};
			  \addplot[orange,only marks,mark=star,mark size=3.5,very   thick] coordinates {
					(0,1)
					(\Nb-1,1)};
				\addplot[only marks,mark=*,mark size=2.5] expression [domain=0:(\Na+\Nb-2), samples=\Na+\Nb-1] {-1};
			\end{axis}
		\end{tikzpicture}
        }
        \subfloat[$\mathbb{S}=\!\!\{0,4,9\}\!\subset\!\mathbb{D}_{\tx}\!=\!
        \mathbb{S}\!\cup\!\{1,6\}$
        ]{\label{fig:sens_sel_gen}
		\begin{tikzpicture} 
			\begin{axis}[width=\wfac*\figsize cm,height=\figsize*0.4 cm,ytick={-1,0,1},
            yticklabels={},
            xmin=-0.2,xmax=\Lmax+0.2,ymin=-1.5,ymax=1.5,xtick={0,1,...,\Lmax},title style={yshift=0 pt},xticklabel shift = 0 pt,xtick pos=bottom,ytick pos=left,axis line style={draw=none}]
				\addplot[blue,only marks,mark=square*,mark size=2.5] expression [domain=0:(\Na-1), samples=\Na] {0};
				\addplot[orange,only marks,mark=star,mark size=3.5,very thick] expression [samples at ={0,\D,...,\D*(\Nbb-1)}] {1};
				
 					\addplot[orange,only marks,mark=star,mark size=3.5,very thick,opacity=0.2] coordinates {
					(1,1)
					(6,1)};
				
				\addplot[only marks,mark=*,mark size=2.5] expression [domain=0:(\Na+\D*(\Nbb-1)-1), samples=\Na+\D*(\Nbb-1)] {-1};
			\end{axis}
		\end{tikzpicture}
        }\vspace{-.2cm}
        \caption{Tx subarrays $\mathbb{S}$ sufficient for satisfying \eqref{eq:sum_set} by \cref{thm:lb}. 
        Including 
        the edge sensors is necessary 
        (\cref{thm:ub}). 
        }
        \label{fig:sens_sel}
    \end{figure}
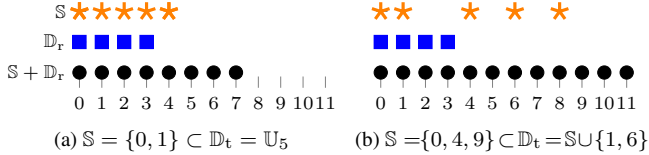

    \begin{remark}[Necessity of redundancy] \label{rem:redundant_im}
    	A redundant array configuration is necessary 
        for communications via Tx sensor selection, 
        when each Tx subarray is constrained to achieve a sum set equal to the (full, possibly contiguous) sum co-array.
    \end{remark}
    
    \cref{rem:redundant_im} hints at an intuitive trade-off between the sizes of the sum co-array $N_\Sigma$ and codebook $|\mathcal{C}^\star|$: 
    a small $N_\Sigma$ seems desirable to achieve a large codebook, while a large $N_\Sigma$ would guarantee identifying more targets. 
    Rigorously establishing and characterizing this trade-off requires deriving both tighter upper and lower bounds for intermediate values of $N_\Sigma$ in \eqref{eq:NSigma_range}. 
    As a step in this direction, we next present 
    a new lower bound displaying this expected trade-off between $|\mathcal{C}^\star|$ and $N_\Sigma$.
    \vspace{-.2cm}
    
    \subsection{Lower bound on $|\mathcal{C}^\star|$: Nested subarray
    construction}
   \begin{proposition}[Lower bound]\label{thm:lb}
   	Suppose $(Q,N_{\tx},N_{\rx},N_\Sigma)$ is admissible and $L=N_\Sigma/N_{\rx}\in\mathbb{N}_+$. Then 
   	\begin{align}
   		|\mathcal{C}^\star(Q,N_{\tx},N_{\rx},N_\Sigma)|\geq \binom{N_{\tx}-L}{Q-L}.\label{eq:lb}
   	\end{align}
   \end{proposition}
   \begin{proof}
   	The proof follows by construction: we choose $\mathbb{D}_{\tx}$ and $\mathbb{D}_{\rx}$ appropriately, and fix $L=N_\Sigma/N_{\rx}$ of the $Q$ sensors in each Tx array subset $\mathbb{S}\subseteq \mathbb{D}_{\tx}$ such that the sum set is $\mathbb{S}+\mathbb{D}_{\tx}=\mathbb{U}_{N_\Sigma}$ regardless of how the remaining $Q-L$ sensors in $\mathbb{S}$ are selected. In particular, let $\mathbb{D}_{\rx}=\mathbb{U}_{N_{\rx}}$ and $\mathbb{D}_{\tx}=\mathbb{D}_1\cup\mathbb{D}_2$, where $\mathbb{D}_1=N_{\rx}\mathbb{U}_L$ and $\mathbb{D}_2$ is an arbitrary $N_{\tx}-L$ element subset of $\mathbb{U}_{N_{\rx}(L-1)+1}\setminus\mathbb{D}_1$. 
   	Any $Q$-sensor subset $\mathbb{S}$, where $\mathbb{D}_1\subseteq \mathbb{S}\subseteq \mathbb{D}_{\tx}$, thus satisfies $\mathbb{U}_{L}= \mathbb{D}_{\tx}+\mathbb{D}_{\rx}\supseteq \mathbb{S}+\mathbb{D}_{\rx}\supseteq \mathbb{D}_1+\mathbb{D}_{\rx}=\mathbb{U}_{L}$. There are at least $\binom{N_{\tx}-|\mathbb{D}_1|}{Q-|\mathbb{D}_1|}$ such unique choices of $\mathbb{S}$, since the remaining $Q-|\mathbb{D}_1|$ sensors of $\mathbb{S}$ can be freely chosen from $\mathbb{D}_2$ without affecting the sum set. The proof is completed by noting that $|\mathbb{D}_1|=L=N_\Sigma/N_{\rx}$. 
   \end{proof}

				
				

    The lower bound \eqref{eq:lb} of \cref{thm:lb} is constructive and hence \emph{achievable}. Specifically, the Rx array is chosen to be a standard ULA $\mathbb{D}_{\rx}=\mathbb{U}_{N_{\rx}}$ and the Tx array (as well as every subset thereof) to contain an $L$ sensor dilated ULA $\mathbb{D}_{\tx}\supseteq\mathbb{S}\supseteq N_{\rx}\mathbb{U}_{L}$. This generalizes the construction establishing tightness of \eqref{eq:ub} for the standard ULA Tx/Rx array in the proof of \cref{thm:ula_nonred}. \cref{fig:sens_sel_gen} shows an example of this generalized construction. Extensions of \cref{thm:lb} to non-integer values of $N_\Sigma/N_{\rx}$ are possible and part of future work.\vspace{-.2cm}

    \subsection{Towards a tighter characterization of $|\mathcal{C}^\star|$}
    
    \cref{fig:bounds} illustrates the derived bounds \labelcref{eq:ub,eq:lb} as a function of the number of sum co-array elements $N_\Sigma$ and selected Tx sensors $Q$. The number of physical sensors $N_{\tx},N_{\rx}$ is fixed in both cases. \cref{fig:bounds_NSigma} shows that for a fixed $Q$, the lower bound on the size of the optimal codebook $|\mathcal{C}^\star|$ decreases with increasing $N_\Sigma$, which is consistent with 
    a trade-off between identifiability and codebook size. The upper bound is tight at $N_\Sigma=N_{\tx}+N_{\rx}-1$ and $N_\Sigma=N_{\tx}N_{\rx}$. For intermediate values $N_{\tx}+N_{\rx}-1<N_\Sigma<N_{\tx}N_{\rx}$, tighter bounds likely exist and present a pertinent direction for future work. 
    
    \cref{fig:bounds_Q} shows that for a fixed $N_\Sigma$, both the upper and lower bounds are concave in $Q$. This raises the question: which value of $Q$ maximizes $|\mathcal{C}^\star|$? For the unconstrained codebook in \eqref{eq:codebook_unconstrained}, the answer is clear since $\binom{N_{\tx}}{Q}=\frac{N_{\tx}!}{(N_{\tx}-Q)!Q!}$ is maximized by 
    $Q=\lfloor N_{\tx}/2 \rceil$. 
A similar conclusion holds for the IM-ISAC codebook $\mathcal{C}^\star$ if $N_\Sigma=N_{\tx}+N_{\rx}-1$ (ULA), since by \cref{thm:ula_nonred}, $|\mathcal{C}^\star|=\binom{N_{\tx}-2}{Q-2}$, which is maximized by $Q=\lfloor N_{\tx}/2-1 \rceil+2 $. However, as \cref{fig:bounds_Q} demonstrates, the answer is nontrivial for other values of $N_\Sigma$, 
    and 
    remains 
    an interesting open question for future work.
    \vspace{-.2cm}

   	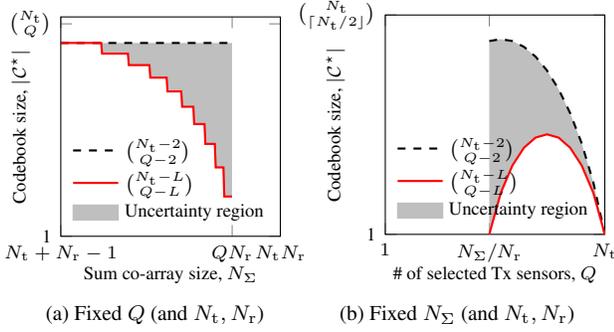
\begin{figure}
   	\centering
   	\subfloat[Fixed $Q$ (and $N_{\tx},N_{\rx}$)]{\label{fig:bounds_NSigma}
   			\begin{tikzpicture}[
   				declare function={binom(\n,\k)=\n!/(\k!*(\n-\k)!);}]
   				\begin{axis}[width=1*\figsize cm,height=1*\figsize cm,yticklabel shift = 0 pt,
   					ylabel shift=-15 pt,
   					xticklabel shift = 0 pt,xlabel shift = -5 pt,xmin
   					=\Nt+\Nr-1,xmax=\Nt*\Nr,ymin=1,
   					xlabel={Sum co-array size, $N_\Sigma$},
   					xtick={\Nt+\Nt-1,\Nr*\Q,\Nr*\Nt},
   					xticklabels={$N_{\tx}+N_{\rx}-1$,$QN_{\rx}$,$N_{\tx}N_{\rx}$},
   					ylabel={Codebook size, $|\mathcal{C}^\star|$}, ytick={1,{binom(\Nt,\Q)}},yticklabels={1,$\binom{N_{\tx}}{Q}$},
   					ymax = {binom(\Nt,round(\Nt/2))},
   					ymode=log,
   					xmode=log,
   					legend style = {
   						at={(0.5,0.02)},anchor = south,legend columns=1,draw=none,fill=none},legend cell align={left}
   					]
   					
   					\addplot[domain=\Nt+\Nr-1:\Nr*\Q,samples=(\Nr*\Q-\Nt-\Nr), 
   					thick,black,dashed,name path=ub]{binom(\Nt-2,\Q-2)};
   					\addlegendentry{$\binom{N_{\tx}-2}{Q-2}$}
   					\addplot[domain=\Nt+\Nr-1:\Nr*\Q,samples=(\Nr*\Q-\Nt-\Nr),
   					thick,red,name path=lb]{binom(\Nt-max(2,floor(x/\Nr)),round(\Q-max(2,floor(x/\Nr))))};
   					\addlegendentry{$\binom{N_{\tx}-L}{Q-L}$}
   					\addplot[color=lightgray,fill=lightgray]
   					fill between[of=ub and lb];
   					\addlegendentry{Uncertainty region}
   					
   				\end{axis}
   			\end{tikzpicture} 
   	}\hspace{-1.5em}%
   	\subfloat[Fixed $N_\Sigma$ (and $N_{\tx},N_{\rx}$)]{\label{fig:bounds_Q}
   			\begin{tikzpicture}[
   				declare function={binom(\n,\k)=\n!/(\k!*(\n-\k)!);}]
   				\begin{axis}[width=1*\figsize cm,height=1*\figsize cm,yticklabel shift = 0 pt,
   					ylabel shift=-25 pt,
   					xticklabel shift = . pt,xlabel shift = -5 pt,xmin
   					=1,xmax=\Nt,ymin=1,
   					xlabel={\# of selected Tx sensors, $Q$},
   					xtick={1,\NSigma/\Nr,\Nt},
   					xticklabels={$1$,$N_\Sigma/N_{\rx}$,$N_{\tx}$},
   					ylabel={Codebook size, $|\mathcal{C}^\star|$}, ytick={1,{binom(\Nt,round(\Nt/2))}},
                    yticklabels={1,$\binom{N_{\tx}}{\lceil N_{\tx}/2\rfloor}$},
   					ymax={binom(\Nt,round(\Nt/2))},
   					ymode=log,
   					legend style = {
   						at={(0.5,0.02)},anchor = south,legend columns=1,draw=none,fill=none},legend cell align={left}
   					]
   					
   					\addplot[domain=ceil(\NSigma/\Nr):\Nt,samples=(\Nt-ceil(\NSigma/\Nr)+1), 
   					thick,black,dashed,name path=ub]{binom(\Nt-2,x-2)};
   					\addlegendentry{$\binom{N_{\tx}-2}{Q-2}$}
   					\addplot[domain=ceil(\NSigma/\Nr):\Nt,samples=(\Nt-ceil(\NSigma/\Nr)+1),
   					thick,red,name path=lb]{binom(\Nt-ceil(\NSigma/\Nr),round(x-ceil(\NSigma/\Nr)))};
   					\addlegendentry{$\binom{N_{\tx}-L}{Q-L}$}
   					\addplot[color=lightgray,fill=lightgray]
   					fill between[of=ub and lb];
   					\addlegendentry{Uncertainty region}
   					
   				\end{axis}
   			\end{tikzpicture} 
   	}\vspace{-.2cm}
   	\caption{Derived bounds on size of the optimal IM-ISAC codebook. The bounds are tight for $N_\Sigma=N_{\tx}+N_{\rx}-1$ (ULA) and $N_\Sigma=N_{\tx}N_{\rx}$ (nonredundant arrays).
    }\label{fig:bounds}
   \end{figure}

    \section{Conclusions}\vspace{-0.3cm}
    
    This paper presented first results on transmit-sensor-selection-based ISAC waveform design with identifiability guarantees. Such waveforms find applications in resource-efficient MIMO ISAC systems, where, for instance, the transmitter has a limited number of RF chains due to power or cost constraints. 
    We formulated a novel codebook optimization (transmit subarray selection) problem, where the size of the codebook was maximized subject to 
    guaranteeing a desired number of identifiable targets. 
    We fully characterized the size of this codebook in case of two widely-considered array geometries: the ULA and nonredundant array. We showed that only redundant arrays can communicate via transmit sensor selection under the considered identifiability constraint. Interestingly, the ULA achieves the maximum codebook cardinality, 
    providing greater constellation size and 
    flexibility in choosing codewords. This comes at the expense of lower target identifiability compared to less redundant arrays---indicative of a general trade-off between communications and sensing performance. 
    To tentatively characterize this trade-off, we derived upper and lower bounds on the size of the optimal codebook. Directions for future work include tightening these bounds and further exploring the potential of sensor-selection-based waveform 
    design in ISAC.
    \vspace{-0.4cm}


    \bibliographystyle{IEEEtran}
    \let\oldthebibliography\thebibliography
	\let\endoldthebibliography\endthebibliography
	\renewenvironment{thebibliography}[1]{
	  \begin{oldthebibliography}{#1}
	    \setlength{\itemsep}{0em}
	    \setlength{\parskip}{0em}
	}
	{
	  \end{oldthebibliography}
	}
    \bibliography{IEEEabrv,references}
	
\end{document}